\documentclass{article}

\usepackage[noadjust]{cite} 
\usepackage{enumerate}
\usepackage{amsmath, amssymb, amsthm,complexity,algorithm,algpseudocode}
\usepackage{euler}
\usepackage{authblk}
\usepackage{todonotes}
\usepackage{subcaption}

\usepackage{mathtools}
\usepackage{thmtools}
\usepackage{thm-restate}
\usepackage{hyperref}
\usepackage{cleveref}
\usepackage{complexity}
\usepackage{authblk}
\theoremstyle{plain}
\usepackage{todonotes}
\usepackage{microtype}
\usepackage{multirow}
\usepackage{amsmath,amssymb}
\usepackage{comment}
\usepackage{enumerate}
\usepackage{xspace}
\usepackage{listings}
\usepackage{color}
\usepackage{graphicx}
\RequirePackage{fancyhdr}
 \usepackage{xcolor}
\usepackage{boxedminipage}
\usepackage{enumitem}

\newcommand{\defparproblem}[4]{
  \vspace{3mm}
\noindent\fbox{
  \begin{minipage}{.95\textwidth}
  \begin{tabular*}{\textwidth}{@{\extracolsep{\fill}}lr} \textsc{#1}\\ \end{tabular*}
  {\bf{Input:}} #2  \\
  {\bf{Parameter:}} #3 \\
  {\bf{Question:}} #4
  \end{minipage}
  }
  \vspace{2mm}
}

\newcommand{\defproblem}[3]{
  \vspace{3mm}
\noindent\fbox{
  \begin{minipage}{.95\textwidth}
  \begin{tabular*}{\textwidth}{@{\extracolsep{\fill}}lr} #1  \\ \end{tabular*}
  {\bf{Input:}} #2  \\
  {\bf{Question:}} #3
  \end{minipage}
  }
  \vspace{2mm}
  }

\def\mypara#1{\sbox0{\parbox{\linewidth}{%
  \parindent0pt
 #1\par\xdef\myparasize{\the\prevgraf}}}%
\ifnum\myparasize=1
{\parindent0pt #1\par}%
\else
#1
\fi}

\newcommand{\AAA}{{\mathcal A}}

\newcommand{\OO}{{\mathcal O}}

\newcommand{\FF}{{\mathcal F}}

\newcommand{\ld}{{\sc Longest $(s,t)$-Detour}}

\newcommand{\sv}[1]{}

\newtheorem{corollary}{Corollary}
\newtheorem{reduction rule}{Reduction Rule}%[subsection]
\newtheorem{branching rule}{Branching Rule}
\newtheorem{definition}{Definition}

\newtheorem{lemma}{Lemma}
\newtheorem{observation}{Observation}

\newtheorem{theorem}{Theorem}
\newtheorem{claim}{Claim}
\usepackage{listings}
\usepackage{color}
\usepackage{complexity}
\usepackage{todonotes}

\title{Long Directed Detours: Reduction to $2$-Disjoint Paths\footnote{It has recently come to our attention that a similar result is given in~\cite{hatzel2023simpler}. Our work was conducted independently of this work.} \footnote{This work was supported by the European Research Council (ERC) grant titled PARAPATH.}}

\date{}
\author[1]{Ashwin Jacob}
\author[1]{Micha{\l} W{\l}odarczyk}
\author[1]{Meirav Zehavi}

\affil[1]{Ben Gurion University of the Negev, Beersheba, Israel
  \texttt{\{ashwinj|meiravze\}@bgu.ac.il}\\
  \texttt{michal.wloda@gmail.com}}

\begin{document}

\maketitle

\begin{abstract}
In the \textsc{Longest $(s,t)$-Detour} problem, we look for an $(s,t)$-path that is at least $k$ vertices longer than a shortest one. We study the parameterized complexity of \textsc{Longest $(s,t)$-Detour} when parameterized by $k$: this falls into the research paradigm of `parameterization above guarantee'. Whereas the problem is known to be fixed-parameter tractable (FPT) on undirected graphs, the status of \textsc{Longest $(s,t)$-Detour} on directed graphs remains highly unclear: it is not even known to be solvable in polynomial time for $k=1$. Recently, Fomin et al. made progress in this direction by showing that the problem is FPT on every class of directed graphs where the \textsc{3-Disjoint Paths} problem is solvable in polynomial time. We improve upon their result by weakening this assumption: we show that only a polynomial-time algorithm for \textsc{2-Disjoint Paths} is required. What is more, our approach yields an arguably simpler proof.
%In the second variant, we are given a graph $G$ with girth $g$, and the objective is to determine whether $G$ contains a path of length at least $g \cdot k$. We show that the problem is in {XP} parameterized by $k$.
\end{abstract}

%\newpage

%\clearpage
%\pagenumbering{arabic}
\section{Introduction}
In the {\sc Longest Path} problem, given a (directed or undirected) graph $G$ and a non-negative integer $k$, the objective is to determine whether $G$ contains a (simple) path of length at least $k$. For over more than three decades, this problem has been studied extensively in the field of parameterized complexity, and has led to the development of various algorithmic techniques such as color-coding, algebraic methods and representative sets~\cite{alon1995color,koutis2008faster,williams2009finding,koutis2015algebraic,fomin2016efficient}.

A recent trend in parameterized complexity is to consider above guarantee versions of  {\sc Longest Path} and its variants where the parameter is the difference between $k$ and some guarantee. Hence, the length of the sought path could be potentially large while the parameter is small.
One such example is the {\ld} problem, where we are also given two vertices $s$ and $t$ of $G$, and the objective is to determine whether $G$ contains a path from $s$ to $t$ of length at least $dist_G(s,t) +k$, where $dist_G(s,t)$ is the length of a shortest path from $s$ to $t$ in $G$ (assuming that one exists) and the parameter is $k$. 

Let $\mathcal{C}$ be the class of directed graphs where {\sc $2$-Disjoint Paths} problem is polynomial-time solvable.
We show that the  {\ld} problem is FPT in directed graphs in $\mathcal{C}$.

%We also consider a ``multiplicative above guarantee" parameterization for {\sc Longest Path} above girth. In {\sc Longest Path Above Girth} problem, given a directed graph $G = (V, E)$ with girth $g$ and a non negative integer $k$, the objective is to determine whether $G$ contains a path of length at least $g \cdot k$. We showed that {\sc Longest Path Above Girth} is {XP} parameterized by $k$.

%\todo[inline]{Previous results and Related Work : other above guarantee parameterizations}

The {\ld} problem was first studied by Bezakova et. al.~\cite{bezakova2017finding}, where they showed that the problem is FPT in undirected graphs. They also proved that the {\sc Exact $(s,t)$-Detour} problem, where we check for an $(s,t)$-path of length exactly $dist_G(s,t)+k$, is solvable in $6.745^kn^{\OO(1)}$ time for directed and undirected graphs. Later, Fomin et. al.~\cite{fomin2022detours} showed that the {\ld} problem is solvable in $45.5^kn^{\OO(1)}$ time for directed graphs in graph classes where the {\sc $3$-Disjoint Paths} problem is polynomial-time solvable. Thus, we improve the requirement of solvability of {\sc $3$-Disjoint Paths} to the solvability of {\sc $2$-Disjoint Paths}. We note that for directed graphs in general, it is not even known whether the problem is in {P} or {NP}-hard for $k=1$~\cite{bezakova2017finding,fomin2022detours,gutin2022survey}.
%\medskip \noindent{\bf Techniques.} The techniques used in this paper to solve {\ld} problem is inspired from the algorithm of Fomin et al.~\cite{fomin2022detours}. Our algorithm uses  the algorithm for {\sc Exact $(s,t)$-Detour} by Bezakova et al.~\cite{bezakova2017finding} 
%over all values $\ell \in [k,2k-1]$. Then it identifies vertices $u,v,x \in V(G)$ and a subset $F \subseteq V(G)$ and constructs a path which is the concatenation of a shortest path from $s$ to $u$, a path of length $k$ from $x$ whose vertices are within $F$ and two vertex disjoint paths from $x$ to $v$ and $v$ to $t$. We then prove that the path constructed is indeed a solution path using arguments based on the lengths of these subpaths.

\medskip \noindent{\bf Related Work.} The paradigm of ``above guarantee" parameterizations of problems was introduced by Mahajan and Raman~\cite{mahajan1999parameterizing}. Since then, such parameterizations were studied for a wide variety of problems, such as {\sc Vertex Cover}~\cite{lokshtanov2014faster,gutin2011vertex}, {\sc Max Sat}~\cite{mahajan1999parameterizing}, {\sc Max Cut}~\cite{alon2011solving,mahajan1999parameterizing} and {\sc Multiway Cut}~\cite{cygan2013multiway}, to name a few.

Other ``above guarantee" versions of the {\sc Longest Path} and {\sc Longest Cycle} problems were studied in the literature for {\em undirected} graphs. Fomin et al.~\cite{fomin2020going} gave an algorithm to find paths (cycles) of length $d+k$ in a connected ($2$-connected) $d$-degenerate graph in $2^{\OO(k)}n^{\OO(1)}$ time. Jansen et al.~\cite{jansen2019hamiltonicity} showed that the {\sc Hamiltonian Cycle} problem can be solved in $2^{\OO(k)}n^{\OO(1)}$ time, if at least $n-k$ vertices have degree at least $n/2$, or if all vertices have degree at least $n/2-k$ (a relaxation of the famous Dirac's Theorem~\cite{dirac1952some}).  For the first parameterization, they did so by giving an $\OO(k)$ kernel. Recently, Fomin et al.~\cite{fomin2022algorithmic} considered a more general relaxation of Dirac's Theorem. They showed that given a $2$-connected graph $G$ and a subset $B \subseteq V(G)$, finding a cycle of length at least $\min\{2\delta(G-B), n-|B|\}+k$ for minimum degree function $\delta$ can be solved in $2^{\OO(k+|B|)} n^{\OO(1)}$ time. Further, Fomin et al.~\cite{fomin2020parameterization} gave an algorithm to find paths and cycles of length at least $g \cdot k$ where $g$ is the girth of the graph in  $2^{\OO(k^2)}n$ time.

%on a $2$-connected graph $G$ in $2^{\OO(k)} n^{\OO(1)}$ time if there exist $n-k$ vertices of degree at least $\delta$. 
\section{Preliminaries}

\medskip \noindent{\bf General Notation.}
Let $\mathbb{N}_0 = \mathbb{N} \cup \{0\}$, where $\mathbb{N}$ is the set of natural numbers. Given $r \in \mathbb{N}_0$, let $[r] = \{1,\ldots,r\}$, and given $r_1, r_2 \in \mathbb{N}_0$,  let $[r_1,r_2] = \{r_1,r_1+1, \ldots,r_2\}$.
Given a finite set $A$ and an integer $t \in \mathbb{N}_0$,  let $2^{A}$ denote the family of all subsets of $A$, and let ${{A}\choose{t}}$ and ${{A}\choose{\leq t}}$ denote the family of all subsets of $A$ size exactly $t$ and at most $t$, respectively. 
%For a function $f:\mathbb{N} \rightarrow \mathbb{N}$, we use $\OO^*(f(k))$ to denote $\OO(f(k)p(n))$ where $p(n)$ is some polynomial in $n$. That is, the $\OO^*$ notation suppresses the polynomial factors in the running times of algorithms.

\medskip \noindent{\bf Graphs.}
We use standard graph theoretic terminology from Diestel's book~\cite{diestel-book}. For a directed or undirected graph $G=(V,E)$, let $V(G)$ denote the set of vertices of $G$, and $E(G)$ the set of edges of $G$. When $G$ is clear from context, let $n = |V(G)|$ and $m = |E(G)|$.  For a set $X \subseteq G$, let $G[X]$ denote the subgraph of $G$ induced on the vertex set $X$, and $G - X$ denote the subgraph of $G$ induced on the vertex set $V(G) \setminus X$. For an edge $e = (u, v)$ in a directed graph, we say that  $v$ is an {\em outneighbor} of $u$ and $u$ is an {\em inneighbor} of $v$. %$u$ is the {\em tail} and $v$ is the head of $e$.

%A {\em directed} graph $G=(V,E)$ is a graph where the edges in $E$ are ordered pairs of vertices from $V$. For an edge of a directed graph with endpoints $u$ and $v$ the endpoint $v$ is said to be {\em adjacent} to $u$. If the graph is undirected,  $u$ and $v$ are adjacent to each other.
 
A {\em walk} (in a given graph $G$) is a sequence of vertices $v_0, \ldots , v_k$ such that for every $i \in \{0, \ldots , k - 1\}$ we have $(v_i, v_{i+1}) \in E(G)$. The walk is {\em closed} if $v_0 = v_k$.
A {\em path} is a walk where no vertex appears more than once. A {\em cycle} is a closed walk where no vertex appears more than once, apart from $v_0$ and $v_k$ (which are the same vertex).
%A {\em path} is a sequence of non-repeating vertices with the property that each vertex in the sequence is adjacent to the vertex preceding it (except the first vertex). 
The {\em length} of a path (or a cycle) $P$  is the number of edges in $P$, and it is denoted by $|P|$. The {\em girth} of a graph $G$ is the shortest length of a cycle in $G$. If the first and last vertices of $P$ are $s$ and $t$, respectively, we call it an {\em $(s,t)$-path}. If an $(s,t)$-path has the shortest length among all the $(s,t)$-paths in the graph $G$, we call the path a {\em shortest} $(s,t)$-path in $G$ (or just shortest $(s,t)$-path if $G$ is clear from context). For a path $P$, let us denote the subpath between any two vertices $u_1$ and $u_2$ of $P$ as $P_{u_1, u_2}$. We say that the walk formed by two paths $P_1$ and $P_2$ where the last vertex of $P_1$  is equal to the first vertex of $P_2$ is the {\em concatenation} of  $P_1$ and $P_2$, denoted as $P_1 \circ P_2$. A vertex $t$ is {\em reachable} from a vertex $s$ in $G$ if there exists a path in $G$ that starts with $s$ and ends with $t$. The {\em distance} between two vertices $s$ and $t$ in $G$, denoted by $dist_G(s,t)$, is the number of edges in a shortest $(s,t)$-path in $G$ (defined as $\infty$ if no such path exists). When $G$ is clear from context, we drop the subscript $G$ from $dist_G(s,t)$.

\medskip \noindent{\bf Parameterized Complexity.}
%\label{prelims-pc}
%\noindent{\textbf{Parameterized Complexity}:} 
A parameterized problem $L$ is a subset of 
$\Sigma^* \times \mathbb{N}_0$ for some finite alphabet $\Sigma$.
Thus, an instance of a parameterized problem is denoted by $(x, k)$ where $x \in \Sigma^*, k \in \mathbb{N}_0$. We denote by $|x|$ the length of $x$.
%We assume that $k$ is given in unary and without loss of generality $k \leq |x|$.

\begin{definition}[Fixed-Parameter Tractability]
\label{defn:fpt}
A parameterized problem $L$ is {\em fixed-parameter tractable} (FPT) if there exists an algorithm $\AAA$, a computable function $f: \mathbb{N} \rightarrow \mathbb{N}$ and a constant $c$ independent of $f$ and $k$, such that given input $(x, k)$, $\AAA$ runs in time $\OO(f(k)|x|^{c})$ and correctly decides whether $(x, k) \in L$.
\end{definition}

%\begin{definition}[XP]
%A parameterized problem $L$ is in {XP} if there exists an algorithm $\AAA$, a computable function $f: \mathbb{N} \rightarrow \mathbb{N}$, such that given input $(x, k)$, $\AAA$ runs in time $\OO(|x|^{f(k)})$ and correctly decides whether $(x, k) \in L$.
%If a parameterized problem $L$ can be solved by an algorithm running in $\OO(n^{f(k)})$ time (called an {XP} algorithm), then $L \in $ {XP}.
%In such a situation we also say that $L$ admits an {XP} algorithm.
%\end{definition}

\medskip \noindent{\bf Pseudorandom Object.}
We will use the following notion.
%\todo[inline]{Splitters}
\begin{definition}[$(n,p,q)$-Lopsided Universal Family]
Let $n, p, q \in \mathbb{N}_0$ such that $n \geq p+q$. A family $\FF$ of sets over a universe $U$ of size $n$ is an {\em $(n,p,q)$-lopsided universal family} if for every $A \in {U \choose \leq p}$ and $B \in {U \setminus A \choose \leq q}$ there exists $F \in \FF$ such that
$ A \subseteq F$ and $B \cap F  = \emptyset$.
\end{definition}

\begin{lemma}[\cite{naor1995splitters}]
\label{lemma:lopsided-universal-family}
 There is an algorithm that given $n, p, q \in \mathbb{N}_0$ such that $n \geq p+q$, constructs an $(n,p,q)$-lopsided universal family $\FF$ of size ${p+q \choose p} \cdot 2^{o(p+q)} \cdot \log n$ in time ${p+q \choose p} \cdot 2^{o(p+q)} \cdot n \log n$.
\end{lemma}

\medskip \noindent{\bf Problem Definitions and Some Known Results.}
We now define the problems relevant to our work. Here, the graph $G$ could be undirected or directed, unless stated otherwise.

\defparproblem{{\sc Longest $(s,t)$-Path}}{A graph $G = (V, E)$, vertices $s,t$ and a non-negative integer $k$.}{$k$}{Is there an $(s,t)$-path in $G$  of length at least $k$?}

We have the following algorithm by Fomin et al.~\cite{fomin2018long} for {\sc Longest $(s,t)$-Path} parameterized by the path length $k$ on directed graphs.

\begin{theorem}[\cite{fomin2018long}]
\label{theorem:longest-path-algo}
{\sc Longest $(s,t)$-Path} on directed (and undirected) graphs can be solved in $4.884^kn^{\OO(1)}$ time.
\end{theorem}

\defparproblem{{\ld}}{A graph $G = (V, E)$, vertices $s,t$ and a non-negative integer $k$.}{$k$}{Is there an $(s,t)$-path in $G$ of length at least $dist_G(s,t)+k$?}

The {\ld} problem was previously studied by Bezakova et al.~\cite{bezakova2017finding}. The authors showed that the problem is FPT on undirected graphs. 

\defproblem{{\sc $p$-Disjoint Paths}}{A graph $G = (V, E)$ and $p$ pairs of terminal vertices $(s_i, t_i), i \in \{1, \dotsc , p\}$ for a non-negative integer $p$.}{ Is there a multiset $\mathcal{P} = \{ P_i : i \in [p],$ $P_i$ is an $(s_i,t_i)$-path in $G\}$ of paths that are vertex disjoint except possibly at their endpoints?}.
Note that, here, two paths can share vertices as long as these vertices are endpoints of {\em both} paths.

Fomin et al.~\cite{fomin2022detours} studied {\ld} on directed graphs and proved the following theorem.

\begin{theorem}[\cite{fomin2022detours}]
\label{theorem:directed-3dp}
Let $\mathcal{C}$ be a class of graphs such that  {\sc $3$-Disjoint Paths} on $\mathcal{C}$ can be solved in polynomial time. Then, {\ld} on $\mathcal{C}$ can be solved in $2^{\OO(k)} n^{\OO(1)}$ time.
\end{theorem}

Bezakova et al.~\cite{bezakova2017finding} also gave an FPT algorithm for {\sc Exact $(s,t)$-Detour}, defined below.

\defparproblem{{\sc Exact $(s,t)$-Detour}}{A graph $G = (V, E)$, vertices $s,t$ and a non-negative integer $k$.}{$k$}{Is there an $(s,t)$-path in $G$ of length exactly $dist_G(s,t)+k$?}

\begin{theorem}[\cite{bezakova2017finding}]
\label{theorem:exact-detour-algo}
{\sc Exact Detour} on directed (and undirected) graphs can be solved in $6.745^k n^{\OO(1)}$ time.
\end{theorem}

\begin{comment}
 Johnson et al.~\cite{johnson2001directed} proved that {\sc Longest Path} is in {XP} parameterized by directed treewidth.

\begin{theorem}[\cite{johnson2001directed}]\label{theorem:longest-path-directed-treewidth}
{\sc Longest Path} on directed graphs is in {XP} parameterized by directed treewidth.
\end{theorem}

We also consider a ``multiplicative above guarantee" parameterization for {\sc Longest Path} above girth. We define the problem below.

\defparproblem{{\sc Longest Path Above Girth}}{A directed graph $G = (V, E)$ with girth $g$ and a non-negative integer $k$.}{$k$}{Is there a path in $G$ of length at least $g \cdot k$?}
\end{comment}

\section{FPT Algorithm for {\ld}}

In this section, we improve upon Theorem \ref{theorem:directed-3dp}: we define $\mathcal{C}$ as the class of graphs where {\sc $2$-Disjoint Paths} (rather than {\sc $3$-Disjoint Paths}) can be solved in polynomial time. Additionally, our algorithm has (arguably) a simpler proof compared to that of Fomin et al.~\cite{fomin2022detours}.

\subsection{Algorithm}
We present the pseudocode of our algorithm in Algorithm \ref{algorithm}. On a high-level, The algorithm calls Theorem \ref{theorem:exact-detour-algo} for {\sc Exact $(s,t)$-Detour} over all values $\ell \in [k,2k-1]$. Then, it iterates over every choice of three vertices $u,v,x \in V(G)$ and subset $F \subseteq \FF$, where $\FF$ is an $(n,k+1,3k+1)$-lopsided universal family over universe $V(G)$. For each such choice, it constructs a path (if one exists) that is the concatenation of:
\begin{enumerate}
    \item a shortest path from $s$ to $u$,
    \item a path of length $k$ from $u$ to $x$ whose vertices are within $F$ and has distance at least $dist_G(s,u)$ from $s$, and
    \item two vertex disjoint paths from $x$ to $v$ and $v$ to $t$.
\end{enumerate}
If such a path does not exist (for all choices), then the algorithm concludes that we have a NO-instance.

%\todo[inline]{Add description of the algorithm in words.}

%Algorithm \ref{algorithm}
\begin{algorithm}
\caption{{\FPT} Algorithm for {\ld} on $\mathcal{C}$}\label{algorithm}
Input: A directed graph $G = (V, E)$, vertices $s,t$, and a non-negative integer $k$.

Output: YES if $G$ has an $(s,t)$-path of length at least $dist_G(s,t)+k$; NO otherwise.

\begin{algorithmic}[1]

\State Check if $t$ is reachable from $s$. If not, return NO.  
\State If $k=0$, return YES.
 \For{$\ell \in [k,2k-1]$} 
	\State 
	\parbox[t]{\dimexpr\linewidth-\algorithmicindent}{	Run the algorithm from Theorem \ref{theorem:exact-detour-algo} for parameter $\ell$. 
	%with $k \leq \ell \leq 2k-1$. If in any case, 
	If the algorithm returns YES, return YES.}
 \EndFor	
\State Compute $R \subseteq V(G)$, the set of vertices reachable from $s$ in $G$, using  breadth-first-search (BFS). Update $G$ to $G[R]$. Let {\em layer} $L_i$, $0 \leq i \leq r$, be the set of vertices at distance $i$ from $s$ in $G$, where $r$ is the maximum distance of a vertex in $G$ from $s$.
\State Construct an $(n,k+1,3k+1)$-lopsided universal family $\FF$ over universe $V(G)$ using Lemma \ref{lemma:lopsided-universal-family}.
 \For{$F \in \FF$} 
%\State \parbox[t]{\dimexpr\linewidth-\algorithmicindent}{Color each vertex in $G$ red or blue with probability $\frac{1}{2}$ independently at random.} 
 \For{$p \in [r]$ and vertices $u,v,x \in V(G)$ where $u,v \in L_p$}

\State \parbox[t]{\dimexpr\linewidth-\algorithmicindent}{Find a shortest $(s,u)$-path $P_1$.}
\State \parbox[t]{\dimexpr\linewidth-\algorithmicindent}{Find a path $P_2$ from $u$ to $x$ of length $k$ in the graph $G[F \cap \bigcup_{i \geq p} L_i]$ (if one exists) using the algorithm in Theorem \ref{theorem:longest-path-algo}.} %Return NO if not found.
\State \parbox[t]{\dimexpr\linewidth-\algorithmicindent}{Check if in $G - (V(P_1) \cup V(P_2) \setminus \{x\})$ there are two vertex disjoint paths, one from $x$ to $v$ and the other from $v$ to $t$, using the polynomial time algorithm for {\sc $2$-Disjoint Paths} on $\mathcal{C}$. If the algorithm returns YES, return YES.}

\EndFor
\EndFor
\State Return NO.
\end{algorithmic}
\end{algorithm}
%\todo[inline]{Change color coding to splitters}

\subsection{Correctness}

%We first set up some notations.

 \noindent{\bf Steps 1-5.} We begin with the following observation.
 \begin{observation}\label{observation:Steps1-5}
 Steps 1-5 of Algorithm \ref{algorithm} are correct.
 \end{observation}
 \begin{proof}
The correctness of Step 1 follows as there is no $(s,t)$-path in $G$ if $t$ is not reachable from $s$. Given that $t$ is reachable from $s$, by the definition of $dist_G(s,t)$, there exists an $(s,t)$-path in $G$ of length $dist_G(s,t)+k$ if $k=0$. Hence, Step 2 of the algorithm is correct as well.

 %In Steps 3-5, we check if is an $(s,t)$-path in $G$ of length $dist_G(s,t)+\ell$ where $\ell \in [k,2k-1]$ using the algorithm in Theorem \ref{theorem:exact-detour-algo}. 
 The correctness of Steps 3-5 follows from Theorem \ref{theorem:exact-detour-algo}.
 \end{proof}

%Suppose the algorithm does not end in Step 5 and the input instance is a YES-instance.  
\medskip \noindent{\bf Steps 6-15.}
From here on, we suppose that the algorithm has reached Step 6. We start with the following observation, which follows as the algorithm did not terminate in Steps 3-5.
%If $P$ is of length $dist_G(s,t)$, the instance is trivially a YES instance from the definition of $dist_G$. 
%From Step 3-5, we can either conclude the correctness from the algorithm of Theorem \ref{theorem:exact-detour-algo} or assume that 

\begin{observation}\label{observation:path-length}
Any  solution path $P$ is of length at least $dist_G(s,t)+2k$.
\end{observation}

%The problem is trivially a YES-instance when $k=0$. Let $R$ be the set of vertices reachable from $s$ in $G$. We can safely assume that the the input instance is $G[R]$ as every solution is contained in this graph as it is an $(s,t)$-path. We can find the set $R$ using a breadth-first-search (BFS) from $s$ in $\OO(n+m)$ time.

%We first run the algorithm from Theorem \ref{theorem:exact-detour-algo} for all parameter values $\ell$ with $k \leq \ell \leq 2k-1$. If in any case, the algorithm returns YES, we stop and return YES as well. Hence we assume that the algorithm returns NO-instance for all such values of $\ell$. Then we know that if the given instance is a  YES-instance, then the path between $s$ and $t$ is of length at least $dist_G(s,t)+2k$.

%Let us now do a BFS from $s$. Let $L_i$ be the set of vertices that are at distance $i$ from $s$. Since the input graph is the set of vertices reachable from $s$, the vertex set of the graph is partitioned into sets $L_i$. We say these sets $L_i$ as layers.

%Suppose that the input instance is a YES-instance. In this case, there exist a path between $s$ and $t$ of length at least $dist_G(s,t) + 2k$. Let $P$ be such a path of the shortest length. 

%Define p,u,v,w wrt P

Next, we identify two vertices in a (hypothetical) solution that will play a central role in the rest of the proof (see Figure \ref{fig:solution-path_uv}).

\begin{figure}
\includegraphics[scale=0.45]{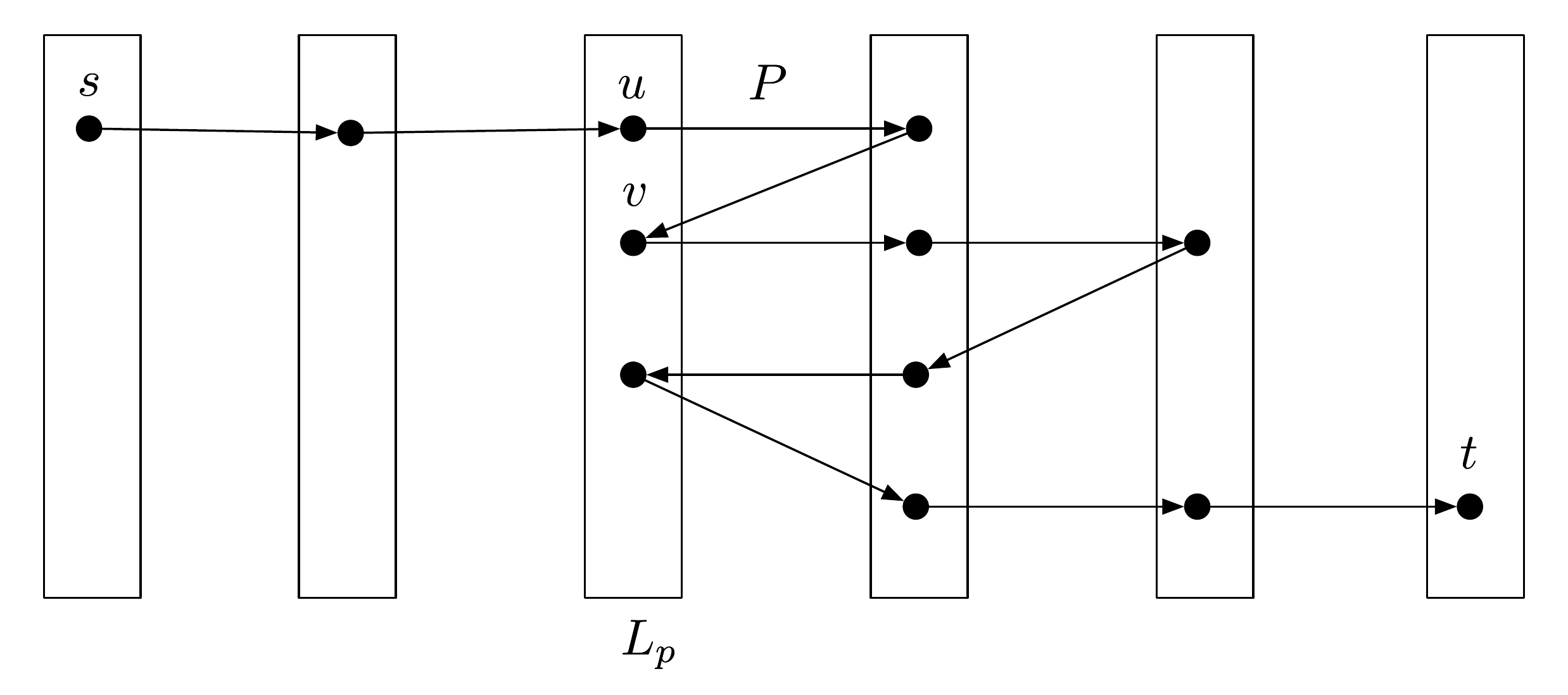}  
\caption{The shortest solution path $P$. Here, $L_p$ is the first layer having multiple vertices of $P$, and vertices $u$ and $v$ are the first two vertices of $P$ in this layer.}\label{fig:solution-path_uv}  
\end{figure}

%\todo[inline]{Label $P$ in the figure.} 
\begin{definition}\label{definition:solution-path}
Let $P$ be a solution path of shortest length.
%Let us look at $P$ with respect to the layers $L_i$ from Step 6. There will be some layers having multiple vertices of $P$ %(not the right reason) since the number of layers is $dist_G(s,t)$
Let $p$ be the smallest integer $i \in [r]$ such that $L_i$ contains more than one vertex of $P$. Let $u$ and $v$ be the first and second vertices of $P$ in $L_p$.\footnote{ 
%We call $L_p$ the {\em leftmost layer with multiple vertices of $P$}.  
Since $t$ is in layer $L_{dist_G(s,t)}$ and $P$ is of length at least $dist_G(s,t) + 2k$ (in particular, strictly larger than $dist_G(s,t)$),  layer $L_p$ exists.}  
\end{definition}

We remark that the algorithm by Fomin et al.~\cite{fomin2022detours} also identifies layer $L_p$ and vertices $u$ and $v$ for the  solution path of shortest length $P$. 
%The algorithm in~\cite{fomin2022detours}  has a polynomial time subroutine using the algorithm for {\sc $3$-Disjoint Paths}.

Note that in one of the iterations of Step 9, integer $p$ and vertices $u$ and $v$ will be the ones defined in Definition \ref{definition:solution-path}. We have the following observation regarding the length of the subpath of $P$ from $s$ to $u$.

\begin{observation}\label{observation:Psu}
Let $P, p, u$ be as defined in Definition \ref{definition:solution-path}. The subpath $P_{s,u}$ of $P$ is a shortest $(s,u)$-path.
\end{observation}
\begin{proof}
Any $(s,u)$-path must contain at least one vertex from each layer $L_i$ with $0 \leq i \leq p$.
Recall that $p$ is the smallest integer $i \in [r]$ such that $L_i$ contains more than one vertex of $P$, and $u$ is the first vertex of $L_p$ in $P$. Thus, the subpath $P_{s,u}$ contains exactly one vertex from each layer $L_i$ with $0 \leq i < p$ and ends at $u \in L_p$. Hence, this subpath is a shortest $(s,u)$-path. 
\end{proof}
Next, we have the following lemma regarding the length of the subpath of $P$ from $u$ to $v$. The lemma statement was also proved by Fomin et al.~\cite{fomin2022detours}. But we give the proof here nevertheless for the sake of completeness.

%Let us look at the case %need to unify this
%where the subpath of $P$ between $u$ and $v$ is of length at most $2k-1$. 

\begin{lemma}\label{lemma:Puv-length}
Let $P,u,v$ be as defined in Definition \ref{definition:solution-path}. The length of the subpath $P_{u,v}$ of $P$ is larger than $k$.
\end{lemma} 
\begin{proof}
%Note that the subpath $P_{s,u}$ is a shortest $(s,u)$-path. This is because $L_p$ is the leftmost layer having multiple vertices of $P$. Hence the subpath $P_{s,u}$ contains exactly one vertex from each layer $L_i$ with $1 \leq i < p$ and ends at $u \in L_p$. Hence this subpath corresponds to a shortest $(s,u)$-path. 

Targeting a contradiction, suppose the subpath $P_{u,v}$ of $P$ has length at most $k$. Then, the length of the subpath $P_{s,v} = P_{s,u} \circ P_{u,v}$, being $|P_{s,u}| + |P_{u,v}|$, is at most $dist(s,u) +k$ ($|P_{s,u}| = dist(s,u)$ from Observation \ref{observation:Psu}). Since $P = P_{s,v} \circ P_{v,t}$, it follows that the length of the subpath $P_{v,t}$, being  $|P| - |P_{s,v}|$, is at least $dist(s,t) + 2k - dist(s,u) - k$. 
 
 Let $P'$ be the walk formed by concatenating an arbitrary shortest $(s,v)$-path $\hat{P}$ with the subpath $P_{v,t}$ of $P$. Note that any shortest $(s,v)$-path consists only of vertices from layers $L_i$ with $0 \leq i < p$ except the vertex $v$. Moreover, by the definitions of $L_p$ and $v$, 
% The vertex $v$ is the second vertex of $P$ in $L_p$ and $L_p$ is the leftmost layer with multiple vertices of $P$. Hence, 
the subpath $P_{v,t}$ consists only of vertices from layers $L_i$ with $i \geq p$. Thus, $P_{v,t}$ is vertex disjoint from $\hat{P}$ (except for the endpoint $v$, which is the only common vertex). Therefore, the walk $P'$ is a path. From Observation \ref{observation:Psu}, we know that the subpath $P_{s,u}$ is a shortest $(s,u)$-path. So, $P = P_{s,u} \circ P_{u,v} \circ P_{v,t}$, $P' = \hat{P} \circ P_{v,t}$ and $|P_{s,u}| = |\hat{P}|$. Hence, the path $P'$ is of length $|P| - |P_{u,v}|$, which is at least $dist(s,t) + 2k - k > dist(s,t) + k$. However, $P'$ is shorter than $P$ as there is at least one edge in the subpath $P_{u,v}$. This contradicts that $P$ is a solution path of shortest length (as defined in Definition \ref{definition:solution-path}).
\end{proof}
%Hence the subpath between $u$ and $v$ has length between $k$ and $2k$. 
%Let us also assume that the length of the subpath of $P$ from $v$ to $t$ is at least $2k$. If this is not the case,  the path from $u$ to $t$ is of size at most $4k$. In this case, we have an algorithm where we guess $u$, use the algorithm from Theorem \ref{theorem:longest-path-algo} to find a path from $u$ to $t$ of length at most $3k$ in $4.884^{4k}n^{\OO(1)}$ time and concatenate it with the shortest path from $s$ to $u$ to get the solution path.
Next, we identify three additional vertices in the (hypothetical) solution path $P$ (see Figure \ref{fig:solution-path-xyz}).
%\todo[inline]{Label $P$ in the figure.} 
\begin{figure}
\includegraphics[scale=0.45]{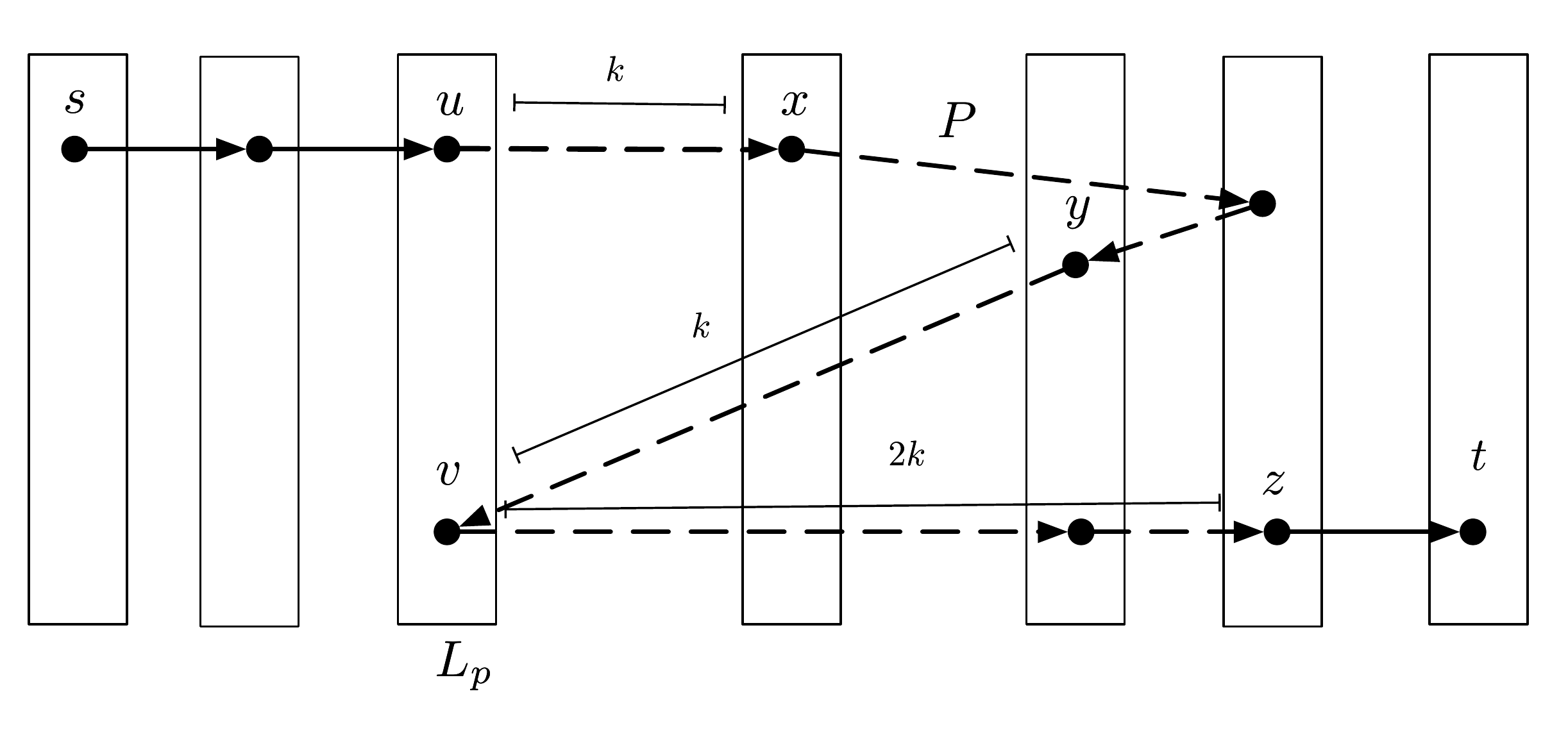}  
\caption{The shortest solution path $P$ with vertices $x,y$ and $z$ as defined in Definition \ref{definition:Pxyz}.}\label{fig:solution-path-xyz}  
\end{figure}

\begin{definition}\label{definition:Pxyz}
 Let $P,u,v$ be as defined in Definition \ref{definition:solution-path}. Let $x$ be last vertex of the subpath induced on the first $k+1$ vertices of $P_{u,v}$ (hence, the length of  $P_{u,x}$ is $k$). Also, let $y$ be the first vertex of the subpath induced on the last $k'$ vertices of $P_{u,v}$ where $k' = \min \{ k+1, |V(P_{x,v})|\}$. Finally, let $z$ be the last vertex of the subpath of  $P_{v,t}$ induced on the first $k''$ verices of $P_{v,t}$ where $k'' = \min \{ 2k+1, |V(P_{v,t})|\}$.
\end{definition}

Note that $x$ exists as the length of $P_{u,v}$ is larger than $k$ from Lemma \ref{lemma:Puv-length}. We remark that in one of the iterations of Step 9, the vertex $x$ will be the one defined in Definition \ref{definition:Pxyz}. Also, $y=x$ when $k' =|V(P_{x,v})|$, and  $z=t$ when $k'' =|V(P_{v,t})|$.

We  have the following lemma regarding $P_{u,x}$ and $P_{y,z}$.
%Where did you use this observation
\begin{observation}\label{observation:lopsided}
Let $\FF$ be the $(n,k+1,3k+1)$-lopsided universal family over universe $V(G)$ constructed in Step 7. Then, 
\begin{itemize}
    \item if $y \neq x$, there exists  $F \in \FF$ such that  $V(P_{u,x}) \subseteq F$ and $V(P_{y,z}) \cap F = \emptyset$.
    \item if $y = x$, there exists $F \in \FF$ such that  $V(P_{u,x}) \subseteq F$ and $(V(P_{y,z}) \setminus \{x\}) \cap F = \emptyset$.
\end{itemize}
\end{observation}
\begin{proof}
For the first case, we have $|V(P_{u,x})| = k+1$ and $|V(P_{y,z})| = |V(P_{y,v})| + |V(P_{v,z})| -1 = k' + k''-1 \leq k+1 + 2k+1-1 =  3k+1$. Moreover, $V(P_{u,x}) \cap V(P_{y,z})= \emptyset$. So, by the definition of $(n,k+1,3k+1)$-lopsided universal family over universe $V(G)$, there exists $F \in \FF$ such that  $V(P_{u,x}) \subseteq F$ and $V(P_{y,z}) \cap F = \emptyset$.
The statement follows for the second case as well as $|V(P_{y,z}) \setminus \{x\}| < |V(P_{y,z})| \leq 3k+1$ and $V(P_{u,x}) \cap (V(P_{y,z}) \setminus \{x\}) = \emptyset$.
\end{proof}
%We now assume that vertices of $P_{u,x}$ are colored blue and those of paths $P_{y,v}$ and $P_{v,z}$ are colored red. The assumption is from using color coding from Steps 4 and 5.

%In Step 5, we colored each vertex of $G$ independently as red or blue with probability $1/2$. The probability that all the vertices of $P_{u,x}$ are colored blue and $P_{y,v}$ and $P_{v,z}$ are colored red is $(\frac{1}{2})^{k+1} (\frac{1}{2})^{k'} (\frac{1}{2})^{k''} \geq \frac{1}{2^{4k+3}}$. Since we repeated the steps independently $\frac{1}{2^{4k+3}}$ times, there is $ (1 - \frac{1}{2^{4k+3}})^{2^{4k+3}} > 1-1/e$ probability that our coloring assumption is true for one of the instances.

%The algorithm guesses the vertices $u,v$ and $x$. It then finds a shortest path from $s$ to $u$, a blue colored path from $u$ to $x$ and delete them. In the remaining graph, it checks for $2$ disjoint paths; one from $x$ to $v$ and the other from $v$ to $t$. If it finds such a path, it returns YES, else NO.

%We now prove that the above algorithm correctly identifies a YES-instance. 
We proceed to classify the edges of any $(s,t)$-path as good or bad (see Figure \ref{fig:good-edge}). Intuitively, an edge is good if it corresponds to the first time the path enters a layer $L_i$ for $i \leq dist(s,t)$.
\begin{figure}
\includegraphics[scale=0.45]{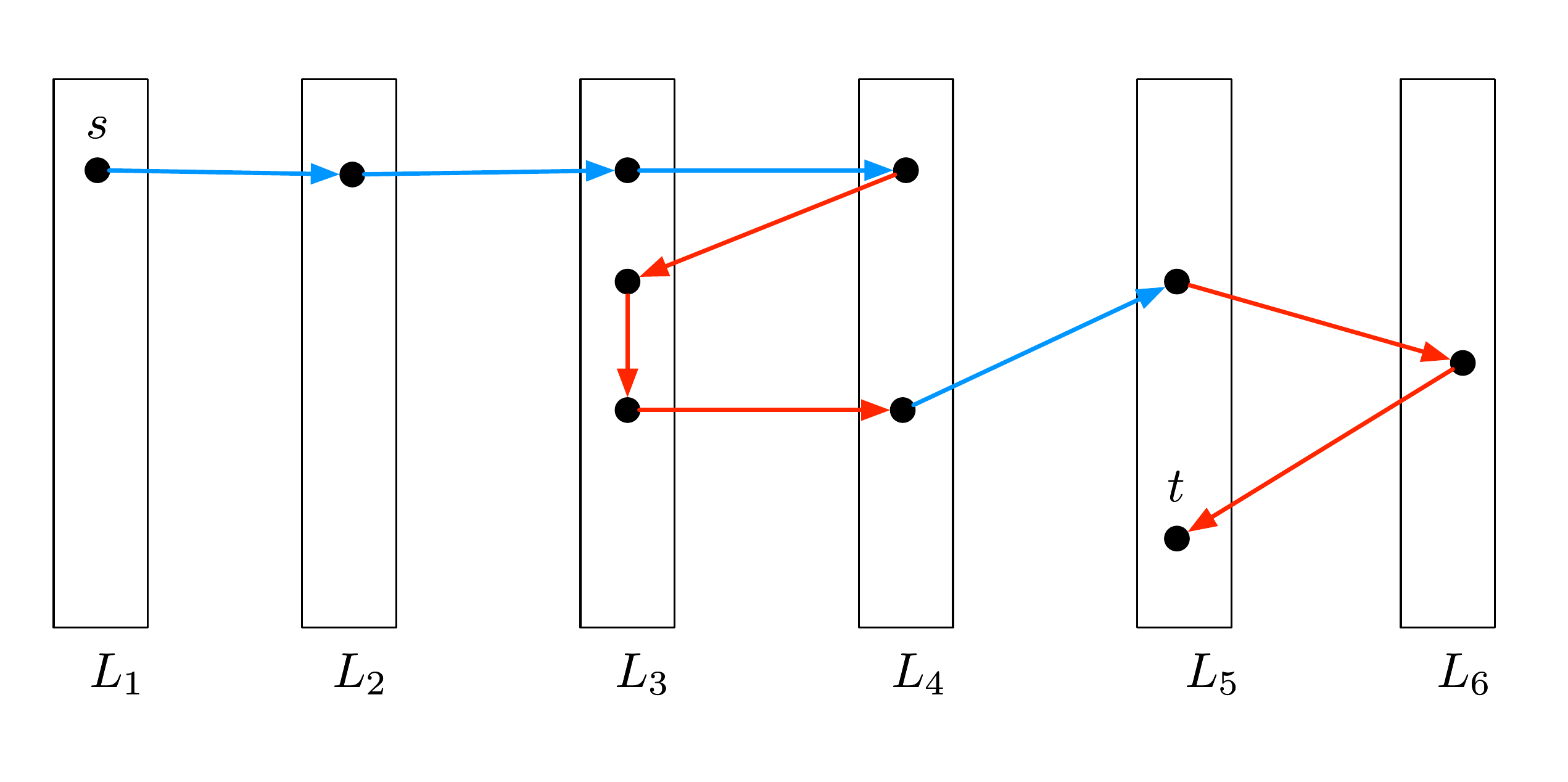}  
\caption{An $(s,t)$-path where the good edges are colored blue and the bad edges are colored red.}\label{fig:good-edge}  
\end{figure}
\begin{definition} Let $\hat{P} = u_1, u_2, \ldots u_q$ be any arbitrary $(s,t)$-path in $G$. An edge $(u_i, u_{i+1})$ of $\hat{P}$, with $1 \leq i \leq q-1$, is {\em good} if $i+1$ is the smallest integer in $[q]$ such that $ u_{i+1} \in L_j$ for some layer $L_j$ with $1 \leq j \leq dist_G(s,t)$; else $(u_i, u_{i+1})$ is {\em bad}.
\end{definition}

 In particular, the edges of any solution path can be partitioned into good edges and bad edges. Intuitively, the good edges can be viewed as the edges between consecutive layers $L_i$ and $L_{i+1}$, with $i <dist_G(s,t)$, that contribute to the distance metric $dist_G(s,t)$. We have the following observation.

\begin{observation}\label{observation:bad-edges}
If an $(s,t)$-path $\hat{P}$ has at least $k$ bad edges, then its length is at least $dist_G(s,t)+k$. 
\end{observation}
\begin{proof}
Let $E(\hat{P}) = E_1 \cup E_2$ where $E_1$ and $E_2$ are the sets of good and bad edges of $\hat{P}$, respectively. We have $|E_2| > k$. By the definition of good edges, there is at least one good edge $(u,v) \in E(\hat{P})$ where $v \in L_j$ for each $1 \leq j \leq dist_G(s,t)$. Hence we have $|E_1| \geq dist_G(s,t)$. Therefore, the length of $\hat{P}$, being $|E_1|+|E_2|$, is at least $dist_G(s,t)+k$.
\end{proof}
Intuitively, the following lemma states that the intersection between the subpath $P_{v,t}$ and any path from $u$ to $x$ of length at most $k$, must belong to $P_{v,z}$ entirely.

\begin{lemma}\label{lemma:disjoint-from-k-layers}
Let $P,u,v,x$ and $z$ be as defined in Definitions \ref{definition:solution-path} and \ref{definition:Pxyz}. Let $P'$ be a path of length at most $k$ from $u$ to $x$ in the graph $G[\bigcup_{i \geq p} L_i]$. 
Then, the subpath $P_{v,z}$ contains all the vertices of $P_{v,t}$ that belong to $P'$.
%Any blue path from $u$ to $x$ is disjoint from the subpath of $P_{v,t}$ from $x$ to $t$.
\end{lemma}
\begin{proof}
  Let $\mathcal{L}'$ be the set of layers that contain at least one vertex in $P'$. Since the path $P'$ is of length at most $k$, $\mathcal{L}' \subseteq \{L_i :  p \leq i \leq p+k\}$. If the subpath $P_{v,t}$ has length at most $2k$, then we are done as $P_{v,z} = P_{v,t}$. Hence, we now suppose that the subpath $P_{v,z}$ of $P_{v,t}$ is of length $2k$. 
  Targeting a contradiction, suppose that $P_{v,t}$ contains a vertex $w \notin V(P_{v,z})$ but $w \in V(P')$. 
  %that is not contained in $P_{v,z}$,  yet it is 
  Thus $w$ is contained in a layer in $\mathcal{L}'$. Choose such a $w$ that is at the closest distance to $v$ in $G$. Since $\mathcal{L}' \subseteq \{L_i :  p \leq i \leq p+k\}$, $w$ belongs to a layer $L_{p+j_1}$ with  $0 \leq j_1 \leq k$. 
  
%  Next, %needproof 
% we show the path $P_{v,t}$ is of length at least $dist_G(v,t) +k$. From Observation \ref{observation:bad-edges}, we can do so by showing 
  We now show that $P_{v,t}$ contains at least $k$ bad edges.
%Without loss of generality, assume that $w$ is the vertex of $P_{v,t}$ that is closest to $v$ that is not in $P_{v,z}$. 
Let $j_2$ be the largest integer such that $L_{p+j_2}$ contains a vertex in $P_{v,z}$.
% $w_1$ be the left neighbor of $w$ in $P$ which is in the layer $L_{p+j_2}$ where $j_2> k$. 
Finally, let $t \in L_{p+j_3}$ for $j_3 \geq 0$ (see Figure \ref{fig:j1-j2-j3-bad} ).
\begin{figure}
$(a)$ \\
\includegraphics[scale=0.45]{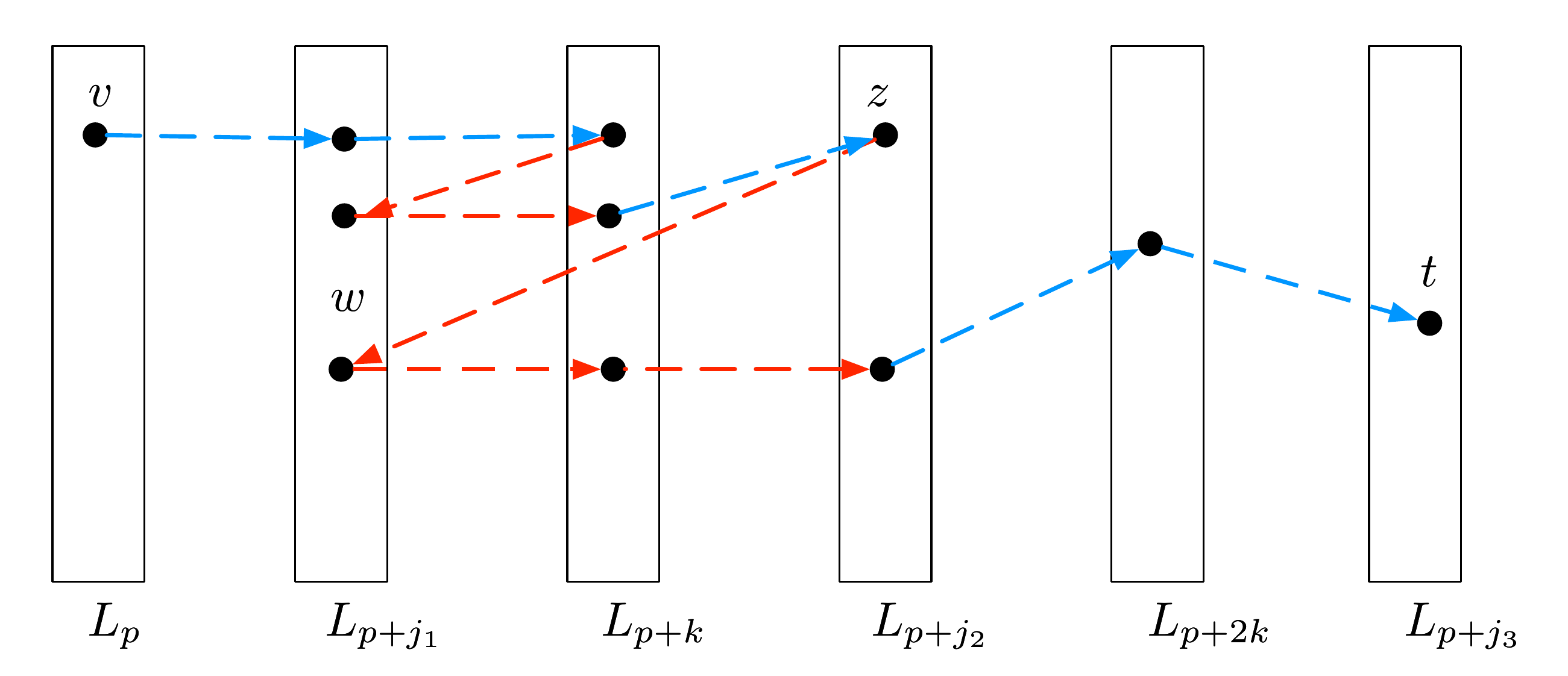}  
$(b)$ \\
\includegraphics[scale=0.45]{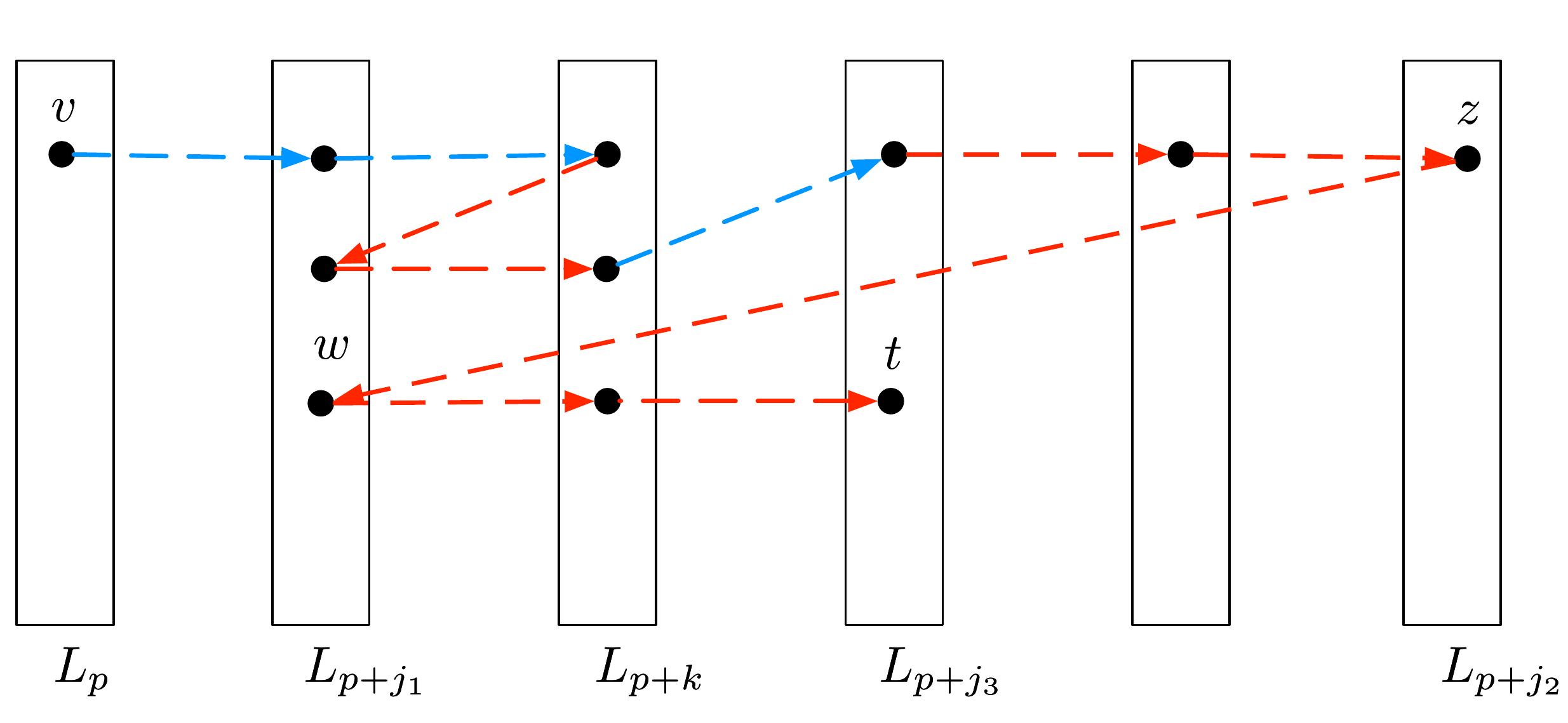} 
\caption{$j_1$, $j_2$ and $j_3$ with the bad edges of $P_{v,t}$ colored in red.}\label{fig:j1-j2-j3-bad} 

\end{figure}

%\todo[inline]{Add (a) and (b). $p+j_2 = p+2k$ in the figure. Will fix.}
Consider the following cases.
%\todo[inline]{Add figure}
%\todo[inline]{Made corrections till here}
%We now define an edge in $P_{v,t}$ as a good edge if its right endpoint is the leftmost vertex of $P_{v,t}$ contained in some layer $L_j$ with $j \leq dist_G(s,t)$. Else we call it as a bad edge. 

%From Observation \ref{observation:bad-edges} and the  definition of $j_2$, 
%We have the following claim.
%We have the following cases.

\begin{comment}
\begin{claim}\label{claim:2k-j2}
There are at least $2k-j_2$ bad edges in $P_{v,t}$.
\end{claim}
%Hence there are $|V(P_{v,z})| = 2k+1$ vertices contained in $j_2+1$ many layers. 
 Since $j_2$ is the largest integer such that $L_{p+j_2}$ contains a vertex in $P_{v,z}$, all the  vertices of $P_{v,z}$ are contained within the set of layers $\{L_i :  p \leq i \leq p+j_2\}$. 
 Recall that an edge is good if its tail vertex corresponds to the first time the path $P_{v,t}$ enters a layer $L_i$. At most $j_2$ of the edges of $P_{v,z}$ are good with its tail vertices corresponding to the first time the path enters a layer $L_i, p+1 \leq i \leq p+j_2$. Thus $|P_{v,z}| - j_2 = 2k-j_2$ edges of $P_{v,t}$ are bad edges. The claim follows.
 
\end{comment}
%To prove that $P_{v,t}$ is of length at least $dist_G(v,t) +k$, we show that there are $k$ bad edges. 

\smallskip \noindent {\bf Case $j_3 \geq j_2$} (See Figure \ref{fig:j1-j2-j3-bad} (a)): 
We first show that there are $2k-j_2$ bad edges of $P_{v,t}$ in the subpath $P_{v,z}$.
Since $j_2$ is the largest integer such that $L_{p+j_2}$ contains a vertex in $P_{v,z}$, all the  vertices of $P_{v,z}$ are contained within the set of layers $\{L_i :  p \leq i \leq p+j_2\}$. 
 Recall that an edge is good if its tail vertex corresponds to the first time the path $P_{v,t}$ enters a layer $L_i$ where $i \leq dist(s,t) = p+j_3$. At most $j_2$ of the edges of $P_{v,z}$ are good with its tail vertices corresponding to the first time the path enters a layer $L_i, p+1 \leq i \leq p+j_2$. Thus $|P_{v,z}| - j_2 = 2k-j_2$ edges of $P_{v,t}$ are bad edges. 

Let us look at the case when $j_2 \leq j_1$. In this case, we have at least $2k-j_2 \geq 2k - j_1 \geq 2k -k = k$ bad edges in the subpath $P_{v,t}$ as $j_1 \leq k$. 

We now show that when $j_2 > j_1$, there are $j_2-j_1$ bad edges of $P_{v,t}$ in the subpath $P_{w,t}$, which is disjoint from $P_{v,z}$. Since $w \in L_{p+j_1}$, $t \in L_{p+j_3} $ and $j_3 \geq j_2$, there is a prefix of the path $P_{w,t}$ from $w$ to a vertex in $L_{p+j_2}$ with all the vertices of the path within the layers $L_i$ with $p+j_1 \leq i \leq p+j_2$. 
%are at least $j_2 - j_1$ edges in the subpath $P_{w,t}$ to  the layer $L_{p+j_2}$. 
All the edges of this prefix path are bad as the path $P_{v,t}$ has already visited these layers earlier in the subpath $P_{v,z}$. 

Hence, overall, we have $2k - j_2 + j_2 - j_1 = 2k - j_1 \geq 2k -k = k$ bad edges in the subpath $P_{v,t}$ as $j_1 \leq k$. 

\smallskip \noindent {\bf Case $j_3 < j_2$} (See Figure \ref{fig:j1-j2-j3-bad} (b)): 
Similarly to the previous case, we first show that there are $2k-j_3$ bad edges of $P_{v,t}$ in the subpath $P_{v,z}$.
%Since $j_2$ is the largest integer such that $L_{p+j_2}$ contains a vertex in $P_{v,z}$, all the  vertices of $P_{v,z}$ are contained within the set of layers $\{L_i :  p \leq i \leq p+j_2\}$. 
 Recall that an edge is good if its tail vertex corresponds to the first time the path $P_{v,t}$ enters a layer $L_i$ where $i \leq dist(s,t) = p+j_3$. At most $j_3$ of the edges of $P_{v,z}$ are good with their tail vertices corresponding to the first time the path enters a layer $L_i, p+1 \leq i \leq p+j_3$. Thus $|P_{v,z}| - j_3 = 2k-j_3$ edges of $P_{v,t}$ are bad edges. 

Let us look at the case when $j_3 \leq j_1$. In this case, we have at least $2k-j_3 \geq 2k - j_3 \geq 2k -k = k$ bad edges in the subpath $P_{v,t}$ as $j_1 \leq k$. 
%In this case, in addition to the $2k-j_2$ bad edges from Claim \ref{claim:2k-j2}, every edge of $P_{v,t}$ whose right endpoint is contained in some layer $L_j$ with $j_3 <  j \leq j_2$ is a bad edge as well. %This includes the edges in $P_{v,z}$. Thus we have $2k - j_2 + j_2 - j_3 = 2k- j_3$ bad edges. 

We now look at the case when $j_3 > j_1$. Since $w \in L_{p+j_1}$ and $t \in L_{p+j_3} $, there are at least $j_3 - j_1$  edges in the subpath $P_{w,t}$. Since $P_{v,z}$ has already visited the layers $L_i$ with $p+j_1 \leq i \leq p+j_3 < p +j_2$, all the edges of $P_{w,t}$ are bad.
Overall, we have $2k - j_3 + j_3 - j_1 = 2k - j_1 \geq 2k -k = k$ bad edges in the subpath $P_{v,t}$ as $j_1 \leq k$. %Hence using the same argument in the previous case, we get a contradiction that $P$ is the shortest solution path as defined in Definition \ref{definition:solution-path}.
\paragraph*{}In both cases, we conclude that there are at least $k$ bad edges in the  subpath $P_{v,t}$. Hence, from Observation \ref{observation:bad-edges}, any path formed by concatenating a shortest $(s,v)$-path with $P_{v,t}$ is of length at least $dist_G(s,t)+k$. This path is shorter than $P$ as there is at least one edge in the subpath $P_{u,v}$. This contradicts that $P$ is a shortest solution path (as defined in Definition \ref{definition:solution-path}).
%We know that $w$ is not the first $2k$ vertices of $P_{v,t}$. Hence 
%  
%  Hence the path obtained by concatenating a shortest path from $s$ to $v$ with $P_{v,t}$ has length at least $dist_G(s,t) +k$ contradicting that $P$ is the shortest such path.
% Suppose not. Since the vertices of the paths $P_{y,v}$ and $P_{v,z}$ are colored red, the path is disjoint from them. 
% 
%Suppose the blue path intersects the non-empty subpath $P_{x,y}$ of $P$ at some internal vertex $w$. Look at the path formed by the shortest path from $s$ to $u$, the blue path from $u$ to $w$, the  subpaths $P_{w,y}$,  $P_{y,v}$ and $P_{v,t}$.  
% 
% 
% Suppose it intersects the subpath $P_{v,t}$ outside these red vertices. Since the path from $u$ to $x$ is of length $k$, it is within the $k$ BFS layers from $u$ to $x$. If it intersects the subpath from $x$ to $t$, the subpath vertex also has to be in the first $k$ BFS layers from $u$. But then we can conclude that there are more than $2k$ vertices from $v$ to $t$ in the first $2k$ layers. 
% 
% Suppose we take the shortest path from $s$ to $v$ and concatenate it with  the subpath $P$ from $v$ to $t$, we still get a path of length $dist_G(s,t)+ k$ as $k$ vertices are repeated in the first $k$ layers from $v$. This contradicts that $P$ is the path corresponding to YES-instance of the shortest length.
\end{proof}

\begin{comment}
 When $P'$ is disjoint from $P_{v,z}$, we have the following corollary of Lemma \ref{lemma:disjoint-from-k-layers}.
 
\begin{corollary}\label{corollary:path-after-z-disjoint}
Let $P,u,v,x$ and $z$ be as defined in Definitions \ref{definition:solution-path} and \ref{definition:Pxyz}. Any path $P'$ from $u$ to $x$ of size at most $k$ that is disjoint from $P_{v,z}$ is disjoint from the subpath $P_{v,t}$ of $P$.
\end{corollary}
\begin{proof}
Let $\mathcal{L}'$ be the set of layers that contain at least one vertex in $P'$.  Lemma \ref{lemma:disjoint-from-k-layers} states that no vertex of $P_{z,t}$ is contained in $\mathcal{L}'$. The proof follows as the intersecting vertex has to be in one of the layers in $\mathcal{L}'$.
\end{proof}
\end{comment}

 The following lemma shows that the subpath $P_{x,y}$ does not intersect any path from $u$ to $x$ of length at most $k$, except for at the vertex $x$.
 
\begin{lemma}\label{lemma:path-xy-disjoint}
 Let $P,u,v,x,y$ and $z$ be as defined in Definitions \ref{definition:solution-path} and \ref{definition:Pxyz}. Any path $P'$ from $u$ to $x$ of size at most $k$ in the graph $G[\bigcup_{i \geq p} L_i]$ that is disjoint from $P_{y,z}$ is disjoint from $P_{x,v}$, except for at the vertex $x$.
\end{lemma}
\begin{proof}
Suppose not. Then, let $w$ be the vertex in $P_{x,v}$ that is closest to $x$ and is contained in $P'$. 
Consider the walk $W$ from $s$ to $t$ defined by $P_{s,u} \circ P'_{u,w} \circ P_{w,y} \circ P_{y,v} \circ P_{v,t}$. If $k' < k$, we have $x=y$, and the claim follows as $P_{x,v} =P_{y,v}$. Hence $|P_{y,v}| = k$. Note that $|P_{s,u}| = |P_{s,v}| = dist_G(s,v)$ as $u,v \in L_p$. So, we have $|P_{s,u}| + |P_{y,v}| + |P_{v,t}| \geq dist_G(s,v) + k + dist_G(v,t) = dist_G(s,t) +k$ as $u,v \in L_p$. The walk $W$ contains the subpaths $P_{s,u}$, $P_{y,v}$ and $P_{v,t}$ which are vertex disjoint except the vertex $v$ which is the last vertex of $P_{y,v}$ and the first vertex of $P_{v,t}$. 
Since $P'_{u,w}$ is disjoint from $P_{v,z}$, by Lemma \ref{lemma:disjoint-from-k-layers}, it is disjoint from $P_{v,t}$. Thus any cycle in the walk $W$ is contained in the sequence $P'_{u,w} \circ P_{w,y}$. Thus, the edges of the paths $P_{s,u}$, $P_{y,v}$ and $P_{v,t}$ are not part of any cycles in the walk $W$. Hence, $W$ contains a path $\hat{P}$ of length that at least the sum of these paths which is at least $dist_G(s,t) +k$. Note that the path $\hat{P}$ is strictly shorter than $P$ as $|P'_{u,w}| < |P'| \leq k = |P_{u,x}|$. This contradicts that $P$ is the shortest solution path as defined in Definition \ref{definition:solution-path}.
%We get a path from $s$ to $t$ by concatenating the shortest $(s,u)$-path, the path from $u$ to $w$ along the path $P'$, the paths $P_{w,v}$ and $P_{v,t}$.
\end{proof}

We have the following corollary from Lemmas \ref{lemma:disjoint-from-k-layers} and \ref{lemma:path-xy-disjoint}.

\begin{corollary}\label{corollary:path-u-to-x-disjoint}
Let $P,u,v,x$ and $z$ be as defined in Definitions \ref{definition:solution-path} and \ref{definition:Pxyz}. Any path $P'$ from $u$ to $x$ of size at most $k$ that is disjoint from $P_{y,z}$ is disjoint from the subpath $P_{x,t}$ of $P$.    
\end{corollary}
\begin{proof}
From Lemma \ref{lemma:disjoint-from-k-layers}, we can conclude that $P'$ is disjoint from $P_{v,t}$ as it is disjoint from $P_{v,z}$. From Lemma \ref{lemma:path-xy-disjoint}, we can conclude that $P'$ is disjoint from $P_{x,v}$, except for at the vertex $x$. Thus, the corollary follows.
\end{proof}
We now prove that Steps 8-15 of the algorithm are correct using the following definition and lemma. First, identify some paths between pairs of  vertices in $u,v,x,y,z$ (see Figure \ref{fig:algorithm-paths}).

\begin{figure}
\includegraphics[scale=0.45]{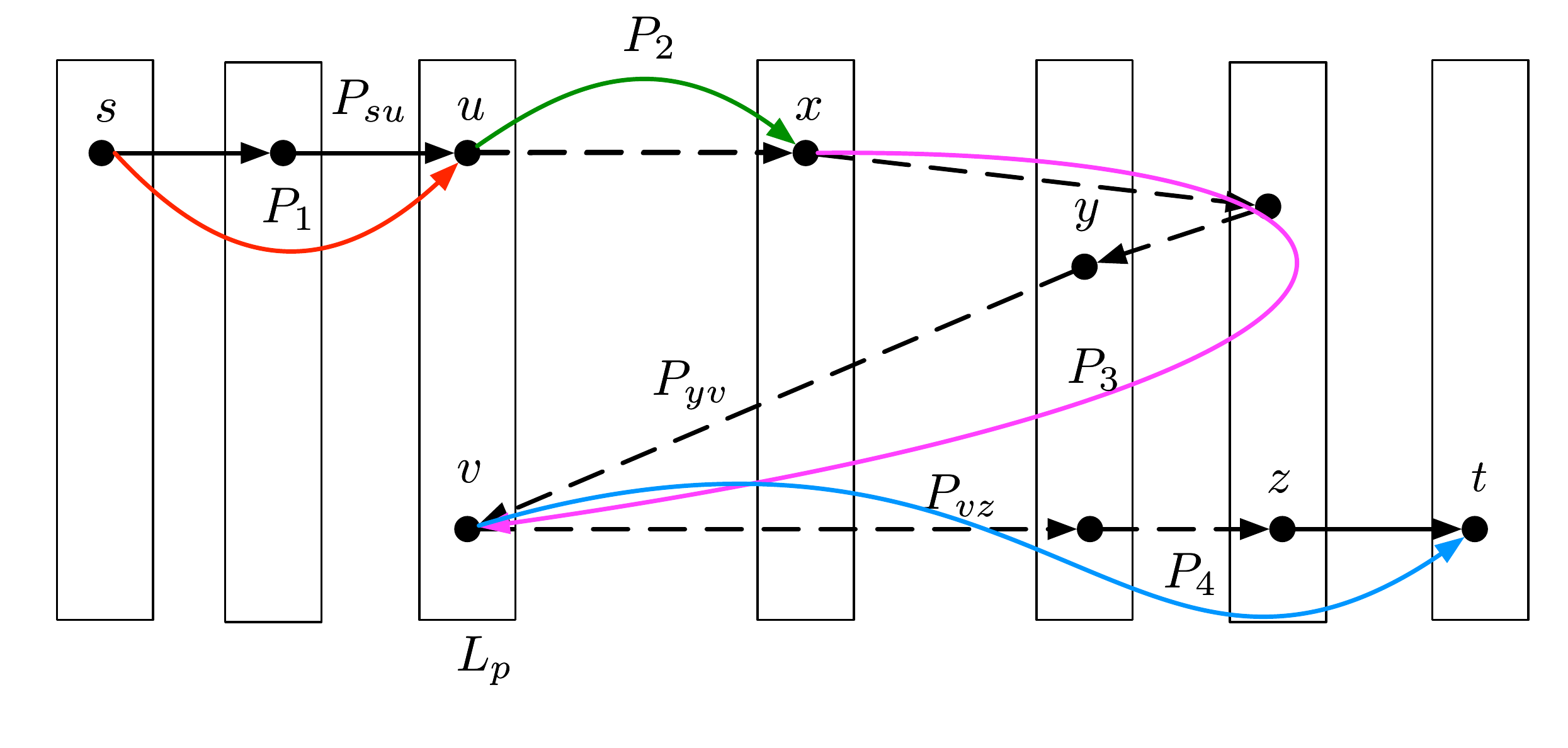}  
\caption{The solution path corresponding to the algorithm in Steps 8-15.}\label{fig:algorithm-paths} 
\end{figure}

%\todo[inline]{Add figure}
\begin{definition}\label{definition:path-conc}
Let $P,u,v,x,y,z$ be as defined in Definitions \ref{definition:solution-path} and \ref{definition:Pxyz}. Let $P_1$ be an arbitrary shortest $(s,u)$-path. Let $P_2$ be a path from $u$ to $x$ of length $k$ that is disjoint from $P_{y,z}$. Let $P_3$ and $P_4$ be internally vertex disjoint paths from $x$ to $v$ and $v$ to $t$,  respectively, in the graph $G - (V(P_1) \cup V(P_2) \setminus \{x\})$.
\end{definition}

We now show that the concatenation of the paths defined above results in a solution path. 

\begin{lemma}\label{lemma:computed-path-length}
 The path $P' = P_1 \circ P_2 \circ P_3 \circ P_4$ is a path of length at least $dist_G(s,t) +k$.
\end{lemma}
\begin{proof}
%First we show that $P'$ is a path. The path $P_1$ is contained in the layers $L_i$ with $1 \leq i \leq p$ with only the last vertex $u$ in layer $L_p$. Also the path $P_4$ is contained in the layers $L_i$ with $i \geq p$ with all the vertices in the layer $L_p$ not $u$. Hence $
The paths $P_3$ and $P_4$ are vertex disjoint from each other and are vertex disjoint from $P_1$ and  $P_2$ by definition. The vertex set of the path $P_1$ is contained in $\bigcup_{i \in [p]} L_i$ with only the last vertex $u$ belonging to layer $L_p$. Also, the path $P_2$ is contained in $\bigcup_{i \geq p} L_i$. Hence, $P_1$ and $P_2$ are internally vertex disjoint as well. We can thus conclude that $P'$ is a path, and that $|P'| = |P_1| + |P_2| + |P_3| + |P_4| \geq dist_G(s,u) + k + |P_3| + dist_G(v,t) \geq dist_G(s,u) + k + dist_G(u,t) = dist_G(s,t) + k$.
\end{proof}

Finally, we prove the correctness of Steps 6-15 of the algorithm by showing that the algorithm indeed computes a collection of paths as defined in Definition \ref{definition:solution-path}.

\begin{lemma}\label{lemma:Steps6-15}
Steps 6-15 of Algorithm \ref{algorithm} are correct.
\end{lemma}
\begin{proof}
%The correctness of Step 6 follows from the fact that any $(s,t)$-path is contained in $G[R]$. Step 7 is correct from Lemma \ref{lemma:lopsided-universal-family}.

We first show that if the instance is a YES-instance, then the algorithm returns YES. Let $P$ be a shortest solution path. In one of the iterations from Step 9, the algorithm considers  $p,u,v,x$ as defined in Definitions \ref{definition:solution-path} and \ref{definition:Pxyz}. In Step 10, the algorithm computes a shortest $(s,u)$-path $P_1$ via BFS. Such a path exists as $P_{s,u}$ is a candidate. In Step 11, it computes a $(u,x)$-path $P_2$ of length $k$ in the graph $G[F]$ for each $F \in \FF$. From Observation \ref{observation:lopsided}, there exists $F \in \FF$ such that $V(P_{u,x}) \subseteq F$ and $V(P_{y,z}) \cap F = \emptyset$ when $y \neq x$. Also, from the same observation, when \item if $y = x$, there exists $F \in \FF$ such that  $V(P_{u,x}) \subseteq F$ and $(V(P_{y,z}) \setminus \{x\}) \cap F = \emptyset$. Then the algorithm will compute $P_2$ as $P_{u,x}$ is a candidate.
%\todo[inline]{The correction said ``and, furthermore". Not sure what to put here.}
%we can conclude that, in one of the iterations of Step 8, we have that $V(P_2) \cap V(P_{y,z}) = \emptyset$. Such a $P_2$ exists as $P_{u,x}$ is a candidate. 
In Step 12, the algorithm computes two vertex disjoint paths, $P_3$ from $x$ to $v$ and $P_4$ from $v$ to $t$ in the graph $G - (V(P_1) \cup V(P_2) \setminus \{x\})$. 
The paths $P_{x,v}$ and $P_{v,t}$ are disjoint from $P_1$ as $P_1$ is contained in layers $L_i$ with $1 \leq i \leq p$ (except $u \in L_p$) and $P_{x,v}$ and $P_{v,t}$ are contained in layers $L_i$ with $i \geq p$. The paths $P_{x,v}$ and $P_{v,t}$ are disjoint from $P_2$ (except the vertex $x$) from Corollary \ref{corollary:path-u-to-x-disjoint} as $P_2$ is disjoint from $P_{y,z}$. Thus, paths $P_3$ and $P_4$ exist as $P_{x,v}$ and $P_{v,t}$ are candidates. 
%since $P_{x,t}$ is disjoint from $P_2$. 
%Now, notice that the paths $P_1, P_2, P_3$ and $P_4$ satisfy the corresponding conditions in Definition \ref{definition:path-conc}. From Lemma \ref{lemma:computed-path-length}, we can conclude that the concatenation of these paths is a solution path as its length is at least $dist_G(s,t)+k$. 
Hence, the algorithm correctly returns YES in Step 12.

We now show that if the algorithm returns YES in Step 12, then the instance is a YES-instance. The algorithm returns YES only when it finds a shortest $(s,u)$-path $P_1$, a path $P_2$ in graph $G[F]$ for some $F \in \FF$ and two vertex disjoint paths $P_3$ and $P_4$, one from $x$ to $v$ and other from $v$ to $t$ in the graph $G - (V(P_1) \cup V(P_2) \setminus \{x\})$. Notice that the paths $P_1, P_2, P_3$ and $P_4$ satisfy the conditions in Definition \ref{definition:path-conc}. From Lemma \ref{lemma:computed-path-length}, we can conclude that the concatenation of these paths is a solution path as its length is at least $dist_G(s,t)+k$. Hence, the instance is a YES-instance.
%If no such path exists for all iterations of Steps 8-9, we correctly return NO in Step 15.
\end{proof}
\subsection{Running Time}
\begin{lemma}\label{lemma:algo1-runtime} Algorithm \ref{algorithm} runs in $46.308^k n^{\OO(1)}$ time.
\end{lemma}
\begin{proof}
Step 1 of the algorithm can be done by a BFS, which takes $\OO(n+m)$ time. In Steps 3-5, we use Theorem \ref{theorem:exact-detour-algo}, which takes $6.745^{2k-1}n^{\OO(1)} \leq 45.495^k n^{\OO(1)}$  time. The BFS in Step 6 takes  $\OO(n+m)$ time. Afterwards, Step 7 takes ${4k+2 \choose k+1} \cdot 2^{o(4k+2)} \cdot n \log n$ time as per Lemma \ref{lemma:lopsided-universal-family}, which is ${4k \choose k}  2^{o(k)}n^{2}$ time. We have
\begin{eqnarray*}
{4k \choose k} &=& \frac{4k!}{k!(4k-k)!}\\
& \leq & \frac{e^{1/12 (4k)} \sqrt{2\pi (4k)} ((4k)/e)^{4k}}{\sqrt{2\pi k}(k/e)^k\sqrt{2\pi(3k)}(3k)/e)^{3k}} \\
&= & \frac{e^{1/12 (4k)}\sqrt{4k}}{\sqrt{2\pi k(3k)}}\left(\frac{(4k)/e}{k/e}\right)^k\left(\frac{(4k)/e}{(3k)/e}\right)^{3k}\\
& \leq & \frac{e^{1/12 (4k)}\sqrt{4k}}{\sqrt{2\pi k(3k)}} 4^k \left(\frac{4}{3}\right)^{3k}\\
& \leq & \left(\frac{256}{27}\right)^{k} 
\end{eqnarray*}
using the Stirling's approximation inequality $n! \leq e^{1/12n} \sqrt{2\pi n} (n/e)^n$
%${4k \choose k} \leq ((4ke/k)^{k} \leq (4e)^{k} $ (using the inequality ${n \choose k} \leq (ne/k)^k$) 
where $e$ is Euler's constant. Thus the time taken by Step 7 is $(256/27)^k 2^{o(k)} n^{2}$.
%which is $\OO({4k \choose k}n^{2})$.  
The size of the $(n,k+1,3k+1)$-lopsided universal family $\FF$ is bounded by ${4k+2 \choose k+1} \cdot 2^{o(4k+2)} \cdot \log n$, which is $(256/27)^k 2^{o(k)} n$. The number of iterations in Steps 8-15 is thus $(256/27)^k 2^{o(k)} n \cdot rn^3$. The shortest $(s,u)$-path in Step 10 can be computed in linear time. The path $P_2$ of size $k$ is computed via the algorithm in Theorem \ref{theorem:longest-path-algo}, which takes $4.884^{k}n^{\OO(1)}$ time. The two disjoint paths in Step 12 are computed by the {\sc $2$-Disjoint Paths} algorithm on class $\mathcal{C}$, which we assume to run in polynomial time. Overall, the time spent in Steps 8-15 is $(256/27)^k \cdot  4.884^{k} 2^{o(k)} n^{\OO(1)}) \leq 46.308^k n^{\OO(1)}$. So, the total running time is $46.308^k n^{\OO(1)}$. 
\end{proof}
%\todo[inline]{Could get better bounds? Stirling's approximation leads to $(256/27)^k \leq 9.4814^k$ it seems from \href{https://www.johndcook.com/blog/2021/09/07/kn-choose-n/}{this website}. Overall running time then is $(9.4814 * 4.884)^k \leq 46.3071^k$ only slightly larger than Fomin et al. algorithm which was $45.495^k$}

We thus have the following theorem, whose proof follows from Observation \ref{observation:Steps1-5},  Lemma \ref{lemma:Steps6-15} and Lemma \ref{lemma:algo1-runtime}.

\begin{theorem}
\label{theorem:directed-2dp}
Let $\mathcal{C}$ be a class of graphs where  {\sc $2$-Disjoint Paths} can be solved in polynomial time. Then, {\ld} on $\mathcal{C}$ can be solved in $46.308^k n^{\OO(1)}$ time.
\end{theorem}

\bibliography{main-file.bib}

\begin{thebibliography}{10}

\bibitem{alon2011solving}
Noga Alon, Gregory Gutin, Eun~Jung Kim, Stefan Szeider, and Anders Yeo.
\newblock Solving max-r-sat above a tight lower bound.
\newblock {\em Algorithmica}, 61(3):638--655, 2011.

\bibitem{alon1995color}
Noga Alon, Raphael Yuster, and Uri Zwick.
\newblock Color-coding.
\newblock {\em Journal of the ACM (JACM)}, 42(4):844--856, 1995.

\bibitem{bezakova2017finding}
Ivona Bez{\'a}kov{\'a}, Radu Curticapean, Holger Dell, and Fedor~V Fomin.
\newblock Finding detours is fixed-parameter tractable.
\newblock In {\em 44th International Colloquium on Automata, Languages, and
  Programming (ICALP 2017)}. Schloss Dagstuhl-Leibniz-Zentrum fuer Informatik,
  2017.

\bibitem{cygan2013multiway}
Marek Cygan, Marcin Pilipczuk, Micha{\l} Pilipczuk, and Jakub~Onufry
  Wojtaszczyk.
\newblock On multiway cut parameterized above lower bounds.
\newblock {\em ACM Transactions on Computation Theory (TOCT)}, 5(1):1--11,
  2013.

\bibitem{diestel-book}
Reinhard Diestel.
\newblock {\em Graph Theory, 4th Edition}, volume 173 of {\em Graduate texts in
  mathematics}.
\newblock Springer, 2012.

\bibitem{dirac1952some}
Gabriel~Andrew Dirac.
\newblock Some theorems on abstract graphs.
\newblock {\em Proceedings of the London Mathematical Society}, 3(1):69--81,
  1952.

\bibitem{fomin2022detours}
Fedor~V Fomin, Petr~A Golovach, William Lochet, Danil Sagunov, Kirill Simonov,
  and Saket Saurabh.
\newblock Detours in directed graphs.
\newblock In {\em 39th International Symposium on Theoretical Aspects of
  Computer Science}, 2022.

\bibitem{fomin2020going}
Fedor~V Fomin, Petr~A Golovach, Daniel Lokshtanov, Fahad Panolan, Saket
  Saurabh, and Meirav Zehavi.
\newblock Going far from degeneracy.
\newblock {\em SIAM Journal on Discrete Mathematics}, 34(3):1587--1601, 2020.

\bibitem{fomin2020parameterization}
Fedor~V Fomin, Petr~A Golovach, Daniel Lokshtanov, Fahad Panolan, Saket
  Saurabh, and Meirav Zehavi.
\newblock Parameterization above a multiplicative guarantee.
\newblock In {\em 11th Innovations in Theoretical Computer Science Conference
  (ITCS 2020)}. Schloss Dagstuhl-Leibniz-Zentrum f{\"u}r Informatik, 2020.

\bibitem{fomin2022algorithmic}
Fedor~V Fomin, Petr~A Golovach, Danil Sagunov, and Kirill Simonov.
\newblock Algorithmic extensions of dirac's theorem.
\newblock In {\em Proceedings of the 2022 Annual ACM-SIAM Symposium on Discrete
  Algorithms (SODA)}, pages 406--416. SIAM, 2022.

\bibitem{fomin2016efficient}
Fedor~V Fomin, Daniel Lokshtanov, Fahad Panolan, and Saket Saurabh.
\newblock Efficient computation of representative families with applications in
  parameterized and exact algorithms.
\newblock {\em Journal of the ACM (JACM)}, 63(4):29, 2016.

\bibitem{fomin2018long}
Fedor~V Fomin, Daniel Lokshtanov, Fahad Panolan, Saket Saurabh, and Meirav
  Zehavi.
\newblock Long directed (s, t)-path: Fpt algorithm.
\newblock {\em Information Processing Letters}, 140:8--12, 2018.

\bibitem{gutin2011vertex}
Gregory Gutin, Eun~Jung Kim, Michael Lampis, and Valia Mitsou.
\newblock Vertex cover problem parameterized above and below tight bounds.
\newblock {\em Theory of Computing Systems}, 48(2):402--410, 2011.

\bibitem{gutin2022survey}
Gregory Gutin and Matthias Mnich.
\newblock A survey on graph problems parameterized above and below guaranteed
  values.
\newblock {\em arXiv preprint arXiv:2207.12278}, 2022.

\bibitem{hatzel2023simpler}
Meike Hatzel, Konrad Majewski, Micha{\l} Pilipczuk, and Marek Soko{\l}owski.
\newblock Simpler and faster algorithms for detours in planar digraphs.
\newblock {\em arXiv preprint arXiv:2301.02421}, 2023.

\bibitem{jansen2019hamiltonicity}
Bart~MP Jansen, L{\'a}szl{\'o} Kozma, and Jesper Nederlof.
\newblock Hamiltonicity below dirac’s condition.
\newblock In {\em International Workshop on Graph-Theoretic Concepts in
  Computer Science}, pages 27--39. Springer, 2019.

\bibitem{koutis2008faster}
Ioannis Koutis.
\newblock Faster algebraic algorithms for path and packing problems.
\newblock In {\em International Colloquium on Automata, Languages, and
  Programming}, pages 575--586. Springer, 2008.

\bibitem{koutis2015algebraic}
Ioannis Koutis and Ryan Williams.
\newblock Algebraic fingerprints for faster algorithms.
\newblock {\em Communications of the ACM}, 59(1):98--105, 2015.

\bibitem{lokshtanov2014faster}
Daniel Lokshtanov, NS~Narayanaswamy, Venkatesh Raman, MS~Ramanujan, and Saket
  Saurabh.
\newblock Faster parameterized algorithms using linear programming.
\newblock {\em ACM Transactions on Algorithms (TALG)}, 11(2):1--31, 2014.

\bibitem{mahajan1999parameterizing}
Meena Mahajan and Venkatesh Raman.
\newblock Parameterizing above guaranteed values: Maxsat and maxcut.
\newblock {\em Journal of Algorithms}, 31(2):335--354, 1999.

\bibitem{naor1995splitters}
Moni Naor, Leonard~J Schulman, and Aravind Srinivasan.
\newblock Splitters and near-optimal derandomization.
\newblock In {\em focs}, page 182. IEEE, 1995.

\bibitem{williams2009finding}
Ryan Williams.
\newblock Finding paths of length k in ${O^*(2^k)}$ time.
\newblock {\em Information Processing Letters}, 109(6):315--318, 2009.

\end{thebibliography}

%\input{Appendix}
%\bibliography{references.bib}
\end{document}